\documentclass[conference]{IEEEtran}
\IEEEoverridecommandlockouts
\usepackage{amsmath,amssymb,amsfonts}%编辑数学公式的宏包
\usepackage{textcomp}
\usepackage{xcolor}    %颜色
\usepackage{wasysym} 
\usepackage{amsthm}  %编辑数学定理和证明过程的宏包
\usepackage{graphicx,times,cite, epsfig, array, mathrsfs}  %插入图片的宏包
\usepackage{array}
\usepackage{multirow}  %复杂表格的宏包
\usepackage[caption=false,font=normalsize,labelfont=sf,textfont=sf]{subfig}
\usepackage{textcomp}
\usepackage{stfloats}
\usepackage{url}
\usepackage{verbatim}
\usepackage{threeparttable}
\usepackage{cases}
\usepackage[linesnumbered, ruled]{algorithm2e}
\SetKwRepeat{Do}{do}{while}%

\usepackage{cite}
\usepackage[all=normal, paragraphs=tight, floats=tight, mathspacing=tight]{savetrees}
\newtheorem{theory}{Theorem}
\newtheorem{lemma}{Lemma}

\begin{document}

\title{Distributed CSMA/CA MAC Protocol for \\ RIS-Assisted Networks}
\vspace{-1cm}
%==============================================================================================================================
\author{ 
\IEEEauthorblockN{
Zhou~Zhang\IEEEauthorrefmark{1},
Saman~Atapattu\IEEEauthorrefmark{2},
Yizhu~Wang\IEEEauthorrefmark{1}, and Marco~Di~Renzo\IEEEauthorrefmark{3}}
\vspace{-0.5cm}
\\
 \IEEEauthorblockA{
 \IEEEauthorrefmark{1}Tianjin Artificial Intelligence Innovation Center, China. \\
 \IEEEauthorrefmark{2}School of Engineering, RMIT University, Melbourne, Victoria, Australia. \\
 \IEEEauthorrefmark{3}Université Paris-Saclay, 3 Rue Joliot Curie, 91190 Gif-sur-Yvette, France. \\
\IEEEauthorblockA{Email:
\IEEEauthorrefmark{1}\{zt.sy1986, wangyizhuj\}@163.com;\,
\IEEEauthorrefmark{2}saman.atapattu@rmit.edu.au;\,
\IEEEauthorrefmark{3}marco.di-renzo@universite-paris-saclay.fr
}
}
% \thanks{This work is supported in part by the National Natural Science Foundation of China under Grant 62171456 and Grant 61801504; and in part by the Australian Research Council (ARC)  Discovery Project under Grant DP220103281 and Future Fellowship under Grant FT210100728. }
% \vspace{-1.0cm}
}

\maketitle

\begin{abstract}
This paper focuses on achieving optimal multi-user channel access in distributed networks using a reconfigurable intelligent surface (RIS). The network includes wireless channels with direct links between users and RIS links connecting users to the RIS. To maximize average system throughput, an optimal channel access strategy is proposed, considering the trade-off between exploiting spatial diversity gain with RIS assistance and the overhead of channel probing. The paper proposes an optimal distributed Carrier Sense Multiple Access with Collision Avoidance (CSMA/CA) strategy with opportunistic RIS assistance, based on statistics theory of optimal sequential observation planned decision. Each source-destination pair makes decisions regarding the use of direct links and/or probing source-RIS-destination links. Channel access occurs in a distributed manner after successful channel contention. The optimality of the strategy is rigorously derived using multiple-level pure thresholds. A distributed algorithm, which achieves significantly lower online complexity at $O(1)$, is developed to implement the proposed strategy. Numerical simulations verify the theoretical results and demonstrate the superior performance compared to existing approaches.

\end{abstract}

\vspace{-0.2cm}
\section{Introduction} \label{s:intro}
The emerging technology of reconfigurable intelligent surfaces (RIS) has the potential to boost wireless network throughput and spectral efficiency \cite{Zdogan2020}. RISs offer cost-effective and energy-efficient advantages over active relays for facilitating efficient transmission among multiple users. However, integrating RISs into multiple user and RIS medium access control (MAC) layer designs faces challenges due to network decentralization and changing channel conditions. This paper focuses on designing reliable and efficient strategies that incorporate RISs to enhance transmission robustness and spectrum utilization efficiency, especially for high-rate data exchange in next-generation smart communications. Extensive research has explored RIS-assisted wireless channel access, primarily focusing on single-RIS system's physical layer design to adjust electromagnetic wave phases for improved communication~\cite{Shuowen2020,Atapattu2020tcom,Fang2022tcom}. RIS integration extended to MIMO systems in \cite{Shuowen2020,Dharmawansa2021wcoml}, demonstrating advantages over traditional relaying methods in \cite{Boulogeorgos2020}. RIS has also been integrated into wireless network MAC design, especially in multi-antenna base station cellular networks as seen in \cite{Wu2019,Gao2021}. The impact of limited RIS phase shifts and size on system performance was analyzed in \cite{HZhang2020}. Further, the transmit powers and the phase shift at each element of the RIS were optimized to maximize the sum-secrecy rate in~\cite{Wijewardena2021coml}.

Current research in this field often assumes the presence of global channel state information (CSI), overlooking the time needed for CSI acquisition. The emerging area of distributed RIS-assisted MAC faces three key challenges: 
(i) Joint scheme for channel contention, CSI acquisition, and RIS-assisted channel access for multiple users.
(ii) Tradeoff between RIS channel acquisition overhead, effective data transmission, channel contention time, and diversity.
(iii) Requirement for low-complexity distributed network operations enabling feasible strategies based on local observations for each user.
{\it This work pioneers the solution to the distributed MAC problem with the assistance of RIS and introduces a robust statistical optimization framework}. 
It significantly contributes to the existing body of research by addressing the following gaps:
\begin{enumerate}
	\item Distributed Network Design: Our novel approach integrates multi-user opportunistic RIS CSI acquisition and assisted access, utilizing optimal sequential observation planned decision theory. This robust framework addresses the MAC problem effectively.  
 \item Optimal RIS-Assisted MAC Strategy: We have developed an optimal MAC strategy that maximizes average throughput by jointly optimizing direct link transmission, opportunistic RIS probing, and RIS-assisted transmission. This statistically proven strategy is practical, supported by closed-form thresholds and a distributed channel access algorithm.
\item
Superior Performance: Comprehensive evaluations and comparisons confirm the superiority of our proposed strategy. It significantly enhances network performance, underlining its potential for improving wireless communication systems. 
\end{enumerate} 
% These contributions collectively address critical research gaps and pave the way for advancements in distributed MAC designs with RIS assistance.
% It bridges the existing research gaps by offering the following contributions:
% \vspace{-0.0cm}
% \begin{enumerate}
% 	\item A distributed wireless network design with a multi-user opportunistic RIS CSI acquisition and assisted access scheme that uses statistics theory of optimal sequential observation planned decision to analyze the MAC problem, which establishes a statistical optimization framework and is the first attempt to address the distributed MAC problem with the aid of RIS.

% 	\item An optimal RIS-assisted MAC strategy that jointly optimizes decisions on direct link transmission, opportunistic RIS probing, and RIS-assisted transmission to maximize the network's average throughput, while ensuring statistical optimality, as rigorously proven.
 
% 	\item A closed-form expression for the thresholds in the proposed strategy and a distributed channel access algorithm that is easy to implement. The proposed strategy is demonstrated to be effective and superior to alternative strategies.
% \end{enumerate} 

{\it Notations}: Throughout this paper, numbers, vectors, and matrices are denoted by lower-case, bold-face lower-case, and bold-face upper-case letters, respectively. 
%The superscript $(\cdot)^T$ and $(\cdot)^H$ stand for the transpose and conjugate transpose, respectively. 
The superscript $(\cdot)^T$ stands for the transpose. $\mathbf{x}\astrosun\mathbf{y}$ denotes the Hadamard product operator.
$\mathbb{E}[\cdot]$ represents the expectation operator. 
%and $\text{Var}[\cdot]$   and variance
Moreover, $\mathcal{CN}(\cdot,\cdot)$ is a circularly symmetric complex Gaussian distribution. For a set $S$, $|S|$ denotes set cardinality.
%\vspace{-0.3cm}
\section{System Model}\label{sub:system_model}
\subsection{Network Model}
\begin{figure}[h]
	\begin{center}
 \vspace{-4mm}
		\includegraphics[scale=.4]{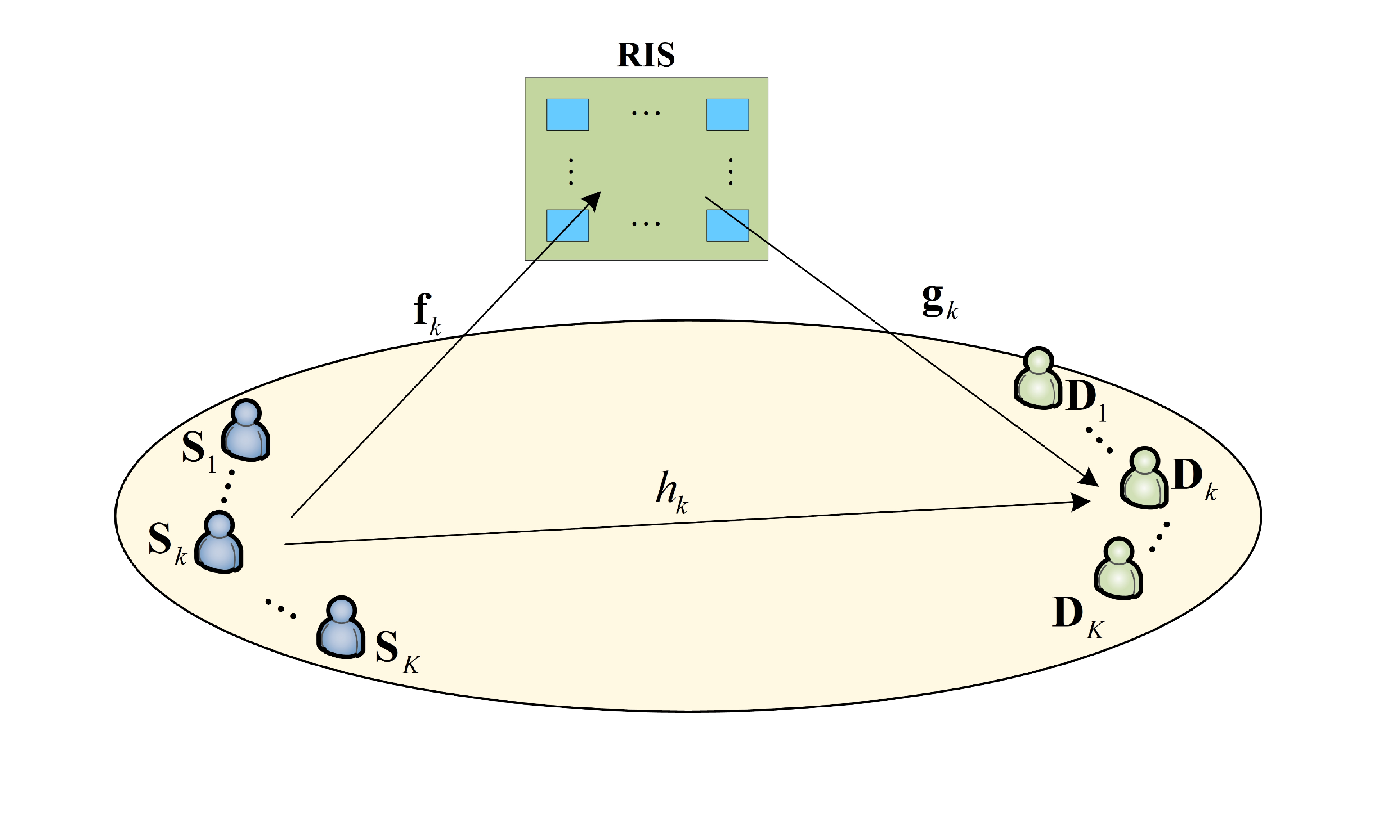}\vspace{-6mm}
		\caption{{A cooperative network with multiple-user pairs and  single RIS.}}\label{f:system_mod}
	\end{center}\vspace{-5mm}
\end{figure}
%\vspace{-0.2cm}
We examine a wireless network comprising a RIS with $M$ passive reflecting elements (REs), as illustrated in Fig.~\ref{f:system_mod}. The network comprises $K$ source-destination pairs denoted as $\text{S}_k$, $\text{D}_k$, and $\text{S}_k-\text{D}_k$, respectively, for $k=1,\cdots,K$. The RIS has its own controller for phase adjustment, and we focus on the transmission of $\text{S}_k-\text{D}_k$ using the $m$th element of the RIS, where $m=1,\cdots,M$. The power budget for each source is $P_t$.
We denote the fading coefficients  of $\text{S}_k-\text{D}_k$ link as $h_k$, the first-hop of the RIS link as $f_{k,m}$, and the second-hop of the RIS link as $g_{k,m}$. We assume independent multipath complex Gaussian fading for all the links, along with large-scale path-loss. 
% Further, the distances between nodes (i.e., sources, destinations, and RIS) are much larger than the sizes of the RIS. Therefore, the statistics of channel gains and distances in a given hop are identical, while the statistics of channel gains and distances of the direct links and two hops are not necessarily identical. 
We denote the distances associated with the direct link of $\text{S}_k-\text{D}_k$ as $d_k$, the first hop and the second hop via the RIS as $d_{k,1}$ and $d_{k,2}$, respectively. 
Thus, we have $\mathbf{f}_{k}=[{f}_{k,m}]\in \mathbb{C}^{M\times 1}$, $\mathbf{g}_{k}=[{g}_{k,m}] \in \mathbb{C}^{M\times 1}$, $h_k\sim {\mathcal{CN}}(0,d_{k}^{-\alpha_1})$, ${f}_{k,m}\sim \mathbf{\mathcal{CN}}(0,d_{k,1}^{-\alpha_2})$, and $g_{k,m}\sim \mathbf{\mathcal{CN}}(0,d_{k,2}^{-\alpha_2})$, where $\alpha_1$ and $\alpha_2$ denote the path-loss exponents for the direct link and the RIS link, respectively. The statistical information of the network, i.e., distances $d_k, d_{k,1}, d_{k,2}, \forall k$, is available at the $\text{D}_k$s.
The additive white Gaussian noise (AWGN) at the RIS and $\text{D}_k$s has noise power $N_0$. For path loss, we denote the antenna gains at each $\text{S}_k$ and $\text{D}_k$ as $G_t$ and $G_r$, respectively, and the reference path loss at $1$\,m distance as $\beta_0$. Channels' coherence time is $\tau_d$.

\vspace{-0.1cm}
\subsection{Framework of the CSMA/CA Protocol}\label{sub:protocol}	
We use the CSMA/CA protocol for  channel access in our RIS-assisted multi-user network~\cite{Mohammad2020,Huang2021,Xuelin2021,Xie2023wcl}. 
The protocol includes request-to-send (RTS) and clear-to-send (CTS), each with time durations $\tau_{R}$ and $\tau_{C}$, respectively. It enables users to initiate transmissions without admission and guarantees that only one user-pair can communicate at a time. % to prevent interference.

In the first stage, $\text{S}_k-\text{D}_k$ pairs contend for channel access over slots with minimum duration $\delta$ using the CSMA/CA protocol. $\text{S}_k$ sends an RTS packet with probability $p_k$ and the MAC protocol proceeds only if a winner is determined, denoted $\text{S}_w-\text{D}_w$. The winner decides to communicate through the direct link, both the direct and RIS links, or not communicate, based on the achievable rate of the direct link, calculated using the direct channel gain $h_w$, and the achievable rate of both links, calculated using the direct channel gain $h_w$ and RIS links' statistics. The direct channel gain $h_w$ is estimated at $\text{D}_w$ upon receiving the RTS, and RIS links' statistical knowledge is assumed to be available at $\text{D}_w$.
The destination $\text{D}_w$ then makes the decision by comparing the achievable rates as follows:
%The destination $\text{D}_w$ then compares the two achievable rates and makes one of three decisions.
\begin{enumerate}
	\item {\tt stop}: If the achievable rate of the direct link is high enough, $\text{D}_w$ selects to use only the direct link, and broadcasts a CTS to all nodes. ${\text{S}}_w$ then transmits data to ${\text{D}}_w$ during the period $(\tau_d-{\tau_{M,1}})$, while all other sources remain silent within the same period. The total time for message exchange is denoted as $\tau_{M,1}=\tau_R+\tau_C$.
	
	\item {\tt continue}: If both  rates are low, $\text{D}_w$ decides not to transmit and broadcasts a CTS to all nodes. This prompts all $\text{S}_k$s to restart channel contention in the next time slot.
	
	\item {\tt assist RIS}: If the achievable rate of the direct link is low but the expected achievable rate of RIS-assisted links is sufficiently high, $\text{D}_w$ delays its decision of whether to {\tt stop} or {\tt continue} until it obtains instantaneous channel knowledge of the RIS link.  	In particular, 		$\text{D}_w$ sends a CTS to the controller of the RIS.
	After receiving the CTS, the RIS controller activates (or switch-on) the RIS. Then $\text{S}_w$ sends RIS training pilots to estimate the instantaneous cascaded channel gain, i.e., $\mathbf{f}_{k}\astrosun\mathbf{g}_{k}$  
	within a time period of $\tau_{p}$, where $\tau_{p}$ denotes the pilot time period~\cite{Qingqing_survey}. This allows $\text{D}_w$ to calculate the instantaneous achievable rate of the direct and RIS-assisted communications. Based on this achievable rate, $\text{D}_w$ makes one of the following decisions:
	\begin{enumerate}
		\item {\tt stop:} If the achievable rate is sufficiently high, $\text{D}_w$ uses RIS-assisted communications and broadcasts a CTS to all nodes. Then $\text{S}_w$ transmits data assisted by the RIS %s in set ${\mathcal J}$ 
		within $(\tau_d-\tau_{M,2})$ time, here $\tau_{M,2}=\tau_{M,1}+\tau_{p}+\tau_{C}$ denotes the total time for pilot and message exchange.
		\item {\tt continue}: If the achievable rate is low, $\text{D}_w$ abandons communications and broadcasts a CTS, prompting all $\text{S}_k$s to contend for the channel in the next time slot. 
	\end{enumerate}
\end{enumerate}
After a successful channel access by {\tt stop}, the next round of channel access among all user pairs initiates the next possible data transmission. 
% \new{
% In the first stage of the network, only the direct links are available, and the RIS remains inactive. $\text{S}_k-\text{D}_k$, $\forall k$, contend for channel access over slots with a minimum duration of $\delta$. At the start of each slot, $\text{S}_k$ sends an RTS packet with probability $p_k$, or keeps silent with probability $(1-p_k)$. The possible contention outcomes are: i) {\it idle}; ii) {\it collision}; or iii) {\it win}. If the outcome is {\it idle} or {\it collision}, sources contend again in the next slot. The MAC protocol proceeds only if a winning pair is determined (e.g., $\text{S}_w-\text{D}_w$) in the third outcome. The winner pair decides to i) communicate solely through the direct link; ii) communicate via both the direct and RIS links; or iii) not communicate due to deep fading in both links. The decision depends on the achievable rate of the direct link, calculated using the direct channel gain $h_w$, and the achievable rate of both RIS-assisted and direct links, calculated using the direct channel gain $h_w$ and RIS links' statistics. The direct channel gain $h_w$ is estimated at $\text{D}_w$ upon receiving the RTS, and the RIS links' statistical knowledge is assumed to be available at $\text{D}_w$. }
The system has two decision levels: the winner source's destination decides to {\tt stop}, {\tt continue}, or {\tt assist RIS} based on direct link rate. If {\tt assist RIS}, the destination evaluates RIS-assisted channels to decide to  {\tt stop} or {\tt continue}.
\vspace{-0.2cm}
\section{Problem formulation} \label{s_DMAC_protocol}
To maximize system throughput, we formulate the problem for RIS-assisted channel access. Using the distributed MAC protocol from Section~\ref{sub:protocol}, we use $w(n)\in{1,...,K}$ to represent the source index of the winning pair after $n$th successful channel contention within one data transmission round. We examine both direct and RIS-assisted links for the winner pair $\text{S}_{w(n)}-\text{D}_{w(n)}$, with the achievable rate of the direct link given as
\begin{align}\label{eq_Rkd}
	R_{w(n),d}=\log_2 \left(1+ \gamma_{w(n),d} \right) 
\end{align}
where $\gamma_{w(n),d}=\bar\rho |h_{w(n)}(n)|^2$, $\bar\rho=P_tG_tG_r\beta_0/N_0$ and 
$h_{w(n)}(n)$ represents its direct link instantaneous CSI.

The matrix that includes the reflecting coefficients of RIS  is denoted as $\mathbf{\Phi}_{w(n)}=\text{diag}([e^{j\phi_{w(n),1}},...,e^{j\phi_{w(n),M}}])\in \mathbb{C}^{M\times M}$, where $\phi_{w(n),m}\in[0,2\pi)$ represents the adjustable angle of the $m$th element of the RIS. To maximize throughput, the optimal reflecting coefficients are given by $\mathbf{\Phi}_{w(n)}^*$ with $\phi_{w(n),m}^*=\arg({h_{w(n)}(n)})-\arg({{f}_{w(n),m}(n)}\cdot{g}_{w(n),m}(n))$, and ${{f}_{w(n),m}(n)}\cdot{g}_{w(n),m}(n)$ represents RIS cascaded link instantaneous CSI for $m=1,\cdots,M$. The achievable rate of the winner pair is then 
\begin{align}\label{equ:2}
	R_{w(n),r}=\log_2(1+\gamma_{w(n),r}(n))
\end{align}
{where} $\gamma_{w(n),r}(n)=\overline{\rho}\Big(|h_{w(n)}(n)|+{|\mathbf{f}_{w(n)}(n)|^T|\mathbf{g}_{w(n)}(n)|}\Big)^2$.

\vspace{-0.2cm}
\subsection{Sequential observation process} 
When a source wins a channel contention and the winning destination estimates its CSI, an {\it observation} (Obs.) occurs. The achievable rate in the first-stage Obs. requires direct link CSI, while the achievable rate in the second-stage Obs. can only be calculated with optimal reflecting coefficients obtained from cascaded CSI. The first-stage Obs. is always performed, while the second-stage Obs. may not be done every time, as discussed in Section~\ref{sub:protocol}. To differentiate between them, odd numbers are used to index the first-stage Obs., and even numbers are used to index the second-stage Obs. Therefore, for the $n$th successful channel contention, two observations are available: Obs.~($2n-1$) and Obs.~($2n$).
For Obs.~($2n-1$), $\text{D}_{w(n)}$ observes $h_{w(n)}(n)$ and the time spent for the $n$th successful channel contention ${t_{w(n)}}(n), n=1,\cdots,\infty$. The number of time slots for a successful channel contention is an i.i.d. geometric random variable~{\cite{Wei2020acm}}. For Obs.~($2n$), $\text{D}_{w(n)}$ observes $\mathbf{f}_{w(n)}(n)\astrosun\mathbf{g}_{w(n)}(n)$. The observed information at Obs.~($2n-1$) and Obs.~($2n$) are denoted as $\mathscr{F}_n=\{w(n),t_{w(n)}(n),h_{w(n)}(n)\}$ and  $\mathscr{G}_{n}=\{\mathbf{f}_{w(n)}(n)\astrosun\mathbf{g}_{w(n)}(n)\}$, respectively. 

\vspace{-0.2cm}
\subsection{Observation path} 
The {\it observation path} $\mathbf{a}$ is the sequence of states for all observations from the first to the $n$th successful channel contention. The path is denoted by odd and even indices, where odd index is always $1$ representing direct-link CSI, and even index represents RIS assist with $a_l=1$ indicating RIS link is estimated and $a_l=0$ otherwise. The observation path $\mathbf{a}$ can be expressed as a vector:
\begin{equation*}%\label{e:path_definition}
\mathbf{a}=\begin{cases}
(1,a_1,\cdots,1,a_l,...,1), & \text{Obs.}(2n-1)\\
(1,a_1,1,\cdots,1,a_l,...,1,a_n), & \text{Obs.}(2n)
\end{cases}
\end{equation*}
% For the observation path $\mathbf{a}$, the observed information is denoted as $\mathscr{B}_{\mathbf{a}}$. 
%Thus we have 
%\begin{equation*}%\label{e:information_definition}
%\mathscr{B}_{\mathbf{a}}=\begin{cases}
%\hspace{-0.1cm}(\lor_{l=1}^n \mathscr{F}_l)\lor\big(\lor_{l=1}^{n-1} \mathscr{G}_{l}\big), & \hspace{-0.3cm}\text{Obs.}(2n-1)\\
%\hspace{-0.1cm}(\lor_{l=1}^n \mathscr{F}_l)\lor\big(\lor_{l=1}^{n} \mathscr{G}_{l}\big), & \hspace{-0.3cm}\text{Obs.}(2n)
%\end{cases}
%\end{equation*}
%where the symbol $\lor_{l=1}^n$ represents  the union of information.
\vspace{-0.3cm}
\subsection{Instantaneous rewards}
For observation path $\mathbf{a}$, we define the reward function $Y_{\mathbf{a}}$ as the transmitted data in {\it bits} by {\tt stop}. For Obs.~($2n-1$) and Obs.~($2n$), we have instantaneous rewards 
\begin{equation}\label{e:reward_definition}
Y_{\mathbf{a}}=\begin{cases}
(\tau_d-\tau_{M,1}) R_{w(n),d}, & \text{Obs.}(2n-1)\\
(\tau_d-\tau_{M,2}) R_{w(n),r}, & \text{Obs.}(2n) \text{ and } a_n\neq0 \\
-\infty, & \text{Obs.}~(2n) \text{ and } a_n=0
\end{cases}
\end{equation}
Additionally, to prevent the winning pair from ending with a decision of {\tt continue}, $Y_{\mathbf{a}}=-\infty$ is defined without any restrictions on feasible MAC strategies. Obs.~($2n-1$) and Obs.~($2n$) are associated with the time cost function $T_{\mathbf{a}}$, which calculates the overall time spent waiting from Obs.~($1$) to Obs.~($2n-1$) or Obs.~($2n$), including data transmission time. Therefore, the function can be formulated as
\begin{equation}\label{e:time_definition}
	T_{\mathbf{a}} = \sum_{l = 1}^n {t_{w(l)}}(l) + \sum_{l = 1}^{n - 1} \mathbb{I}[a_{l}\neq0]({\tau_{C}+{\tau_p}})  + {\tau_d}-\tau_{M,1}.
\end{equation}
\vspace{-0.5cm}
\subsection{Optimization goal} 
We define $\mathbf{a}_s$ as the observation path leading to a {\tt stop} decision in the sequential decision-making process of the distributed MAC strategy described in Section~\ref{sub:protocol}. The transmitted data in {\it bits} and the time cost of strategy $\mathbf{a}_s$ are denoted as $Y_{\mathbf{a}_s}$ and $T_{\mathbf{a}_s}$, respectively, both of which are random variables due to the stochastic nature of the sequential observation process. By repeating strategy $\mathbf{a}_s$ over sufficient rounds of data transmission, the sample average system throughput approaches the average system throughput expressed as ${\mathbb{E}[{Y_{\mathbf{a}_s}}]}/{\mathbb{E}[{T_{\mathbf{a}_s}}]}$. Our objective is to maximize the average system throughput by identifying the optimal strategy $\mathbf{a}_s^*$ among all feasible channel access strategies $\mathbf{a}_s$ and the maximal average system throughput, which can be evaluated, respectively, as
\begin{equation}\label{equ:original_problem}
	\mathbf{a}_s^*=\arg\sup_{\mathbf{a}_s > 0} \frac{{\mathbb{E}[{Y_{\mathbf{a}_s}}}]}{{\mathbb{E}[{T_{\mathbf{a}_s}}]}} \text{ and } \lambda^*=\frac{{\mathbb{E}[{Y_{\mathbf{a}_s^*}}}]}{{\mathbb{E}[{T_{\mathbf{a}_s^*}}]}}.
\end{equation}

\section{Optimal MAC Strategy and  Algorithm}\label{optimal_strategy1}

In this section, our objective is to determine the optimal MAC strategy $\mathbf{a}_s^*$  that maximizes the average system throughput $\lambda^*$ as in \eqref{equ:original_problem}.
%=======================================================
\vspace{-0.3cm}
\subsection{Optimal Strategy} 
Using equivalent transformation and optimal sequential observation planned decision theory, we obtain the optimal MAC strategy $\mathbf{a}_s^*$, which maximizes system throughput~$\sup_{\mathbf{a}_s > 0} {\mathbb{E}[{Y_{\mathbf{a}_s}}]}/{\mathbb{E}[{T_{\mathbf{a}_s}}]}$. The optimal MAC strategy, as described in Theorem~\ref{th:optimal_rule1}, is presented below.

\begin{theory}\label{th:optimal_rule1}
	For a round of successful data transmission, an optimal RIS-aided MAC strategy $\mathbf{a}_s^*$ is in the form that: beginning from $n=1$, after $n$th successful channel contention, the $w(n)$-th source-destination pair wins the channel and $\text{D}_{w(n)}$ obtains CSI $h_{w(n)}(n)$.
	\begin{enumerate}
		\item if $(\tau_d-\tau_{M,1}) (R_{w(n),d}(n)-\lambda^*)\ge \max\{\Lambda_{w(n)}(\lambda^*,|h_{w(n)}(n)|),0\}$, $\text{D}_{w(n)}$ {\tt stop}, and $\text{S}_{w(n)}$ transmits data to $\text{D}_{w(n)}$ in direct link. Function $\Lambda_{w(n)}(\lambda^*,|h_{w(n)}(n)|)$ representing the expected maximal reward if RIS link is estimated after the $n$th successful channel contention, is expressed as
		\begin{align}\label{equ:lambda_def}
				& \Lambda_{w(n)}  (\lambda^*,|h_{w(n)}(n)|)=\mathbb{E}\big[\max\big\{(\tau_d-\tau_{M,2})\nonumber\\
    &\quad			R_{w(n),r}(n)			
			-\lambda^*(\tau_d-\tau_{M,1}),-\lambda^*(\tau_p+\tau_C)\big\}\big].
		\end{align}
		The maximal average system throughput $\lambda^*$
	%$\sup\limits_{\mathbf{a}_s > 0} \frac{{\mathbb{E}[{Y_{\mathbf{a}_s}}}]}{{\mathbb{E}[{T_{\mathbf{a}_s}}]}}$, and 
	can be uniquely solved by equation
	\begin{align}\label{equ:bellman1}
		\frac{1}{K}\sum\limits_{k=1}^K\int\limits_{0}^{+\infty}
		\max & \Big\{(\log_2(1+\overline{\rho}h_s^2)-\lambda^*)(\tau_d-\tau_{M,1}),\nonumber\\
  &\Lambda_{k}(\lambda^*,h_s),0
		\Big\} d F_{|h_{k}|}(h_s)  =
		\lambda^*\tau_o
	\end{align}
	where $F_{|h_{k}|}(h_s)$ denotes the c.d.f. of $|h_{k}|$ which is Rayleigh distributed, and $\tau_o$ denotes the average duration of a successful channel contention $t_w(n)$, expressed as 
		$\tau_o=\tau_{M,1}+ \prod\limits_{k=1}^K(1-p_k)\frac{\delta}{p_s} +\big(1-\prod\limits_{k=1}^K(1-p_k)-p_s\big)\frac{\tau_{R}}{p_s}$,
		with $p_s=\sum\limits_{k=1}^Kp_k\prod\limits_{i\neq k}(1-p_i)$. 	
		\item if $\max\big\{(\tau_d-\tau_{M,1}) (R_{w(n),d}(n)-\lambda^*),\Lambda_{w(n)}(\lambda^*,|h_{w(n)}(n)|)\big\}\!<\!0$, $\text{D}_{w(n)}$ {\tt continue}, gives up transmission and then all sources re-contend the channel.
		\item otherwise, $\text{D}_{w(n)}$ {\tt assist RIS}, and then estimates cascaded RIS channel gains $\mathbf{f}_{w(n)}(n)\astrosun$ $\mathbf{g}_{w(n)}(n)$.
		\begin{enumerate}
			\item if reward $R_{w(n),r}(n)\ge \lambda^*$, $\text{D}_{w(n)}$ {\tt stop}, and $\text{S}_{w(n)}$ transmits data assisted by the RIS with reflecting coefficients matrix $\mathbf{\Phi}_{w(n)}^*$.
			\item otherwise, $\text{D}_{w(n)}$ {\tt continue}, gives up transmission and then all sources re-contend the channel.
		\end{enumerate}
	\end{enumerate}
	\end{theory}
	\begin{proof}
		%As the original problem of maximizing network throughput in (\ref{equ:original_problem}) is analytically challenging, we simplify it by converting it to an alternative problem that uses a price-based objective function. 
		%We use a sequential planned decision (SPD) strategy to determine the optimal MAC strategy. 
		
		%\subsection{Equivalent transformation}
		Let $Z_{\mathbf{a}}(\lambda)=Y_{\mathbf{a}}-\lambda T_{\mathbf{a}}$, where $\lambda$ is the price charged for time spent. A strategy that achieves $\sup\limits_{\mathbf{a}_s> 0}\mathbb{E}[Z_{\mathbf{a}_s}(\lambda)]$ for a given $\lambda>0$ is denoted as $\mathbf{a}_s(\lambda)$, and the corresponding optimal strategy is denoted as $\mathbf{a}_s^*(\lambda)$, which is expressed as  
		\begin{equation} \label{equ:tranfer_prob}
			\mathbf{a}_s^*(\lambda)=\arg\sup\limits_{\mathbf{a}_s> 0}Z_{\mathbf{a}_s}(\lambda)=\arg\sup\limits_{\mathbf{a}_s>0}\{Y_{\mathbf{a}_s}-\lambda T_{\mathbf{a}_s}\}.
		\end{equation}
	As problem (\ref{equ:tranfer_prob}) is the alternative problem which is equivalent to \eqref{equ:original_problem} when $\lambda=\lambda^*$ according to Theorem~1 in \cite{fergusonoptimal},
		the optimal strategy $\mathbf{a}_s^*$ %for our original throughput maximization problem 
		is $\mathbf{a}_s^*(\lambda^*)$, %where $\lambda^*$ is unique and satisfies $\sup\limits_{\mathbf{a}_s> 0}\mathbb{E}[Z_{\mathbf{a}_s}(\lambda^*)]=0$. 
  and is solved by following steps.
  %To solve the problem, we 
		
%		i) find the optimal strategy $\mathbf{a}_s^*(\lambda)$ that achieves $\sup\limits_{\mathbf{a}_s> 0}\mathbb{E}[Z_{\mathbf{a}_s}(\lambda)]$ for a given price $\lambda>0$; ii) replace $\lambda$ with $\lambda^*$; and iii) the strategy $\mathbf{a}_s^*(\lambda^*)$ is the solution $\mathbf{a}_s^*$ to the original problem.
		
		\subsubsection{Find the optimal strategy $\mathbf{a}_s^*(\lambda)$ for price $\lambda>0$} 
		%that achieves $\sup\limits_{\mathbf{a}_s> 0}\mathbb{E}[Z_{\mathbf{a}_s}(\lambda)]$ a given 
		%Next, we focus on finding an optimal strategy to achieve $\sup\limits_{\mathbf{a}_s> 0}\mathbb{E}[Z_{\mathbf{a}_s}(\lambda)]$. 
		%The relation between any two histories is defined as follows: 

		%The set $A_{\mathbf a}$ represents the possible values of the next action following path $\mathbf a$, and it can be defined as \com{where do you use this?}:
		%\begin{equation}\label{e:Aaexp}
		%	A_{\mathbf a} = \left\{
		%	\begin{array}{ll}
			%		\{\mathcal{J}:\mathcal{J}\subset \{1,2,...,L\}\}, & \text{odd $|\mathbf a|$}\\
			%		\{1\}, & \text{even $|\mathbf a|$}
			%	\end{array} \right..
		%\end{equation}
		Based on Theorem~2.14 in \cite{Schmitz_N_book}, % and the definitions above, 
  for any price $\lambda>0$, an optimal strategy $\mathbf{a}_s^*(\lambda)$ is that: Starting from $n=1$, 
  it takes the following procedure until {\tt stop}. In particular,
		\begin{itemize}
			\item at Obs.~$(2n-1)$, it is optimal to {\tt stop} with $\mathbf{a}_s^*(\lambda)=\mathbf{a}$ if $Z_{\mathbf{a}}\ge V_{\mathbf{a}}^1$, and {\tt continue} observation otherwise. If {\tt continue} is optimal, we update the observation path by appending the action as $\mathbf{a}=(\mathbf{a},a_n^*)$ where $\mathbb{E}\big[U_{(\mathbf{a},a_n^*)}|\mathscr{B}_{\mathbf{a}}\big]\!=\!V_{\mathbf{a}}^1$, $a_n^*\in\{0,1\}$, and the observed information for the observation path $\mathbf{a}$ is denoted as $\mathscr{B}_{\mathbf{a}}$;
			\item at Obs.~($2n$), it is optimal to {\tt stop} with $\mathbf{a}_s^*(\lambda)=\mathbf{a}$ if $Z_{\mathbf{a}}\ge V_{\mathbf{a}}^2$, and {\tt continue} observation otherwise. If {\tt continue}, update observation path by adding $1$, i.e.,  $\mathbf{a}=(\mathbf{a},1)$.
		\end{itemize}
		For the strategy, the reward function for observation path $\mathbf{a}$ is defined as $U_{\mathbf{a}} \overset{}{=}\sup\limits_{\mathbf{b}\ge {\mathbf{a}}}\mathbb{E}\big[Y_{\mathbf{b}}-\lambda T_{\mathbf{b}}\big|\mathscr{B}_{\mathbf{a}}\big]$\footnote{For paths $\mathbf{a}$ and $\mathbf{b}$, $\mathbf{b}\ge\mathbf{a}$ if and only if $|\mathbf{b}|\ge |\mathbf{a}|$ and path $\mathbf{a}$ is the initial segment of path $\mathbf{b}$. Here $|\cdot|$ represents the cardinality.}, which represents the expected maximal reward based on the observed information $\mathscr{B}_{\mathbf{a}}$ so far. 
		For Obs.~($2n-1$), we define the reward function $V_{\mathbf{a}}^1\overset{}{=}\max\limits_{j=0,1}\mathbb{E}\big[U_{(\mathbf{a},j)}\big|\mathscr{B}_{\mathbf{a}}\big]$, which represents the expected maximal rewards if {\tt stop} is not made at Obs.~($2n-1$). For Obs.~($2n$), we define the reward function $V_{\mathbf{a}}^2\overset{}{=}\mathbb{E}\big[U_{(\mathbf{a},1)}\big|\mathscr{B}_{\mathbf{a}}\big]$, which represents the expected maximal rewards if {\tt stop} is not made at Obs.~($2n$). If $|\mathbf{a}|=0$ (i.e., no observation has been made yet), we calculate $U_{\mathbf{a}}$ as $U_0=\sup\limits_{\mathbf{a}_s> 0}\mathbb{E}[Z_{\mathbf{a}_s}(\lambda)]$.
		
\subsubsection{Replace $\lambda$ with $\lambda^*$} 
	
	%As the original and transferred problem are equivalent when $\lambda=\lambda^*$, the maximal average throughput $\lambda^*$ can be solved such that $\sup\limits_{{\mathbf{a}_s}>0}\big\{\mathbb{E}[Y_{\mathbf{a}_s}]-\lambda^* \mathbb{E}[T_{\mathbf{a}_s}]\big\}=0$. 
By replacing $\lambda$ with $\lambda^*$, to implement the optimal strategy $\mathbf{a}_s^*(\lambda^*)$, we calculate $U_{\mathbf a}$, $V_{\mathbf a}^1$ and $V_{\mathbf a}^2$ %for path ${\mathbf a}$ 
using the Bellman equation %described 
in Lemma~2.8 in \cite{Schmitz_N_book} as:
%			\begin{align}\label{e:expres3}
%
%		\end{align}
	 \begin{align}
	 	&			U_{\mathbf a}=\max\big\{(\tau_d-\tau_{M,1}) (R_{w(n),d}(n)-\lambda^*),\nonumber\\&~~~~~~~~~~~
	 	\Lambda_{w(n)}(\lambda^*,|h_{w(n)}(n)|),0
	 	\big\}	\!-\!\lambda^* T_c(n)\label{e:expres3}\\
		&V_{\mathbf{a}}^1=
		\max\big\{\Lambda_{w(n)}(\lambda^*,|h_{w(n)}(n)|),0\big\}-\lambda^* T_c(n)\label{equ:def_V1}
		% (\sum\limits_{l = 1}^n {{t_s}} (l) + \sum\limits_{l = 1}^{n - 1} {\mathbb{I}[{a_{l}}=1]{\tau_p}})
		\\&V_{\mathbf{a}}^2=-\lambda^* T_c(n)-\lambda^*\mathbb{I}{[|a_{n}|>0]}(\tau_{p}+\tau_{C})\label{equ:V2}
		%\max\big\{,\Lambda_{s(n),1}(\lambda,|h_s(n)|,U_0)\big\}-
	\end{align}
where $T_c(n) {=}\sum\limits_{l=1}^n t_{w(l)}(l) +\sum\limits_{l=1}^{n\!-\!1}\mathbb{I}{[a_{l}>0]}(\tau_{p}+\tau_{C})$.

By making expectation on both sides of \eqref{e:expres3} with $|\mathbf{a}|=0$, we can solve $\lambda^*$ using the equation $\mathbb{E}\big[\max\big\{(\tau_d\!-\!\tau_{M,1})(R_{w(n),d}(n)\!-\!\lambda^*),\Lambda_{w(n)}(\lambda^*,|h_{w(n)}(n)|),0\big\}\big] \!=\!\lambda^*\tau_o$.
%\begin{align}\label{e:expres4}
%	\mathbb{E}\big[\!\max\big\{\!(\tau_d\!-\!\tau_{M,1}) & (R_{w(n),d}(n)\!-\!\lambda^*),\nonumber\\
%	& \Lambda_{w(n)}(\lambda^*,|h_{w(n)}(n)|),0\!\big\}\big] \!=\!\lambda^*\tau_o.
%\end{align}
	%Based on Bellman equation (\ref{e:expres4}), $\lambda^*$ can be solved by the equation \eqref{equ:bellman1}, and
This can be rewritten as \eqref{equ:bellman1}, and the uniqueness of $\lambda^*$ follows from the monotonicity of its both sides.% of \eqref{equ:bellman1}.	
	
%	\begin{align*}%\label{e:expres3}
%		\mathbb{E}\big[\!\max\big\{\!(\tau_d\!-\!\tau_{M,1})(R_{w(n),d}(n)\!-\!\lambda^*),0,\Lambda_{w(n)}(\lambda^*,|h_{w(n)}(n)|,0)\!\big\}\big] \!=\!\lambda^*\tau_o.
%	\end{align*}
%		
		
	\subsubsection{Solution $\mathbf{a}_s^*$ and decision conditions} 
	
	Since the optimal solution to the original problem is $\mathbf{a}_s^*=\mathbf{a}_s^*(\lambda^*)$, we analyze the optimal decision after each observation under the MAC framework in Section~\ref{sub:protocol}.
	First, we consider the case at Obs.~$(2n-1)$. By definitions in (\ref{e:reward_definition}), (\ref{e:time_definition}) and 
 (\ref{equ:def_V1}), the condition to {\tt stop} is rewritten as $(\tau_d-\tau_{M,1})	(R_{w(n),d}(n)\!-\!\lambda^*)\ge \max\big\{\Lambda_{w(n)}(\lambda^*,|h_{w(n)}(n)|),0\big\}$. 
	% \begin{equation}\label{equ2:1st_stop}
	% 	(\tau_d-\tau_{M,1})	(R_{w(n),d}(n)\!-\!\lambda^*)\ge \max\big\{\Lambda_{w(n)}(\lambda^*,|h_{w(n)}(n)|,0),0\big\}.
	% \end{equation}
	Similarly, 	%by (\ref{e:reward_definition}), (\ref{e:time_definition}) and (\ref{equ:def_V1}), 
 the condition to {\tt continue}  is rewritten as $\max\big\{(\tau_d-\tau_{M,1}) (R_{w(n),d}(n)-\lambda^*),\Lambda_{w(n)}(\lambda^*,|h_{w(n)}(n)|)\big\}
		< 0$.
	% \begin{align*}
	% 	&\max\big\{(\tau_d-\tau_{M,1}) (R_{w(n),d}(n)-\lambda^*),\Lambda_{w(n)}(\lambda^*,|h_{w(n)}(n)|,0)\big\}
	% 	< 0.
	% \end{align*}
  %Otherwise, it is optimal to continue observation, i.e. {\tt RIS assists} with $a_n^*=\argmin_{|a_n|}\{a_n\subset\mathcal{L}: \mathbb{E}\big[U_{(\mathbf{a},a_n)}|\mathscr{B}_{\mathbf{a}}\big]\!=\!V_{\mathbf{a}}^1\}$.
	%by (\ref{eq_iterate1}) and (\ref{def:gamma}), 
	Then the condition to {\tt assist RIS} can be derived.	
%	$\mathcal{J}_{n}^*=\arg\min\limits_{\mathcal{J}\in \mathscr{J}_n^*}|\mathcal{J}|$ where 
%	\[	\mathscr{J}_n^*=\big\{{\mathcal J}: \Lambda_{w(n),{\mathcal J}}(\lambda^*,|h_{w(n)}(n)|,0)=\max\limits_{\mathcal{J}'\subset{\mathcal L}}\Lambda_{w(n),{\mathcal J}}(\lambda^*,|h_{w(n)}(n)|,0)\big\}.\]	
	Secondly, we consider the case at Obs.~$(2n)$. 
	%$\text{D}_{w(n)}$ decides {\tt stop} if $Y_{\mathbf{a}}-\lambda T_{\mathbf{a}}\ge V_{\mathbf{a}}^2$, and {\tt continue} otherwise. 
	By definition in (\ref{equ:V2}), the condition for {\tt stop} is rewritten as $R_{w(n),r}(n)\ge\lambda^*$. %(\ref{e:reward_definition}), (\ref{e:time_definition}) and 
	% \begin{equation*}%\label{equ2:2nd_stop}
	% 	(\tau_d-\tau_{M,2})R_{w(n),r}(n)-\lambda^*(\tau_d-\tau_{M,1})
	% 	\ge -\lambda^*(\tau_p+\tau_C).
	% \end{equation*}
	%It is equivalent to $R_{w(n),\mathcal{J}_n^*,r}(n)\ge \lambda^*$.	
Finally, following the MAC framework in Section~\ref{sub:protocol}, by substituting above conditions into strategy $\mathbf{a}_s^*$, the optimal MAC strategy is derived.
	\end{proof}
\vspace{-0.3cm}
Theorem~1 provides an optimal MAC strategy for the network, requiring efficient calculation of  $\lambda^*$ and  $\Lambda_{w(n)}(\lambda,|h_{w(n)}(n)|)$ in (\ref{equ:bellman1}) and (\ref{equ:lambda_def}), respectively. The former can be calculated offline by solving (\ref{equ:bellman1}), while the latter requires online computation for each instantaneous CSI. We next propose a closed-form solution for the threshold function, enabling efficient calculation for the optimal strategy.
\vspace{-0.2cm}
\subsection{A close-form expression for $\Lambda_{k}(\lambda,h_s)$}
\vspace{-0.1cm}
%As the average throughput $\lambda^*$, set ${\mathcal{K}}^*$ and solutions $\zeta_{w(n)},\eta_{w(n)}, w(n)\in{\mathscr K}^*$ are obtained based on reward functions $\Lambda_{w(n)}(\lambda,h_s)$, we focus on deriving a close-form expression. 

For winner pair $w(n)=k$, by using (\ref{equ:2}) and (\ref{equ:lambda_def}), we have
\vspace{-0.1cm}
\begin{align}\label{equ:ergodic_cap}
	&\Lambda_{k}(\lambda,h_s)\stackrel{\text{(a)}}{\approx} 
	%=(\tau_d-\tau_{M,2})\mathbb{E}\big[\max\big\{\log_2(1+\gamma_{k,r}),\lambda\big\}\big]	-\lambda(\tau_d-\tau_{M,1}) \nonumber\\
		(\tau_d-\tau_{M,2})\log_2\big(1+\mathbb{E}\big[\max\big\{\overline{\rho}\big(h_s \nonumber\\&
  \qquad\quad\quad+|\mathbf{f}_{k}(n)|^T|\mathbf{g}_{k}(n)|\big)^2,2^\lambda-1\big\}\big]\big) -\lambda(\tau_d-\tau_{M,1}) \nonumber\\
	& 	\stackrel{\text{(b)}}{\approx} (\tau_d-\tau_{M,2})		\log_2\big(1+\Omega(\lambda,h_s,\mu_{k},\sigma_{k})\big)
	- \lambda(\tau_d-\tau_{M,1})\nonumber\\ &\triangleq \overline{\Lambda}_{k}(\lambda,h_s)
\end{align}
where (a) is derived by applying the Taylor series, while (b) is obtained by taking the expectation utilizing the Central Limit Theorem (CLT), i.e., $|\mathbf{f}_{k}(n)|^T|\mathbf{g}_{k}(n)|\sim\mathcal{N}(\mu_{k},\sigma_{k}^2)$ with $\mu_{k} =\frac{M\pi }{4}d_{k,1}^{-\alpha_2/2} d_{k,2}^{-\alpha_2/2}$ and $\sigma_{k}^2= M(1-\frac{\pi^2}{16})d_{k,1}^{-\alpha_2} d_{k,2}^{-\alpha_2}$, and  
\begin{align}	&\Omega(\lambda,h_s,\mu,\sigma) %\nonumber\\& 
\!=\!
	 \frac{(2^\lambda\!-\!1)\left(\text{erf}\left(\frac{\mu }{\sqrt{2} \sigma }\right)\!-\!\text{erf}\left(\frac{-\sqrt{\frac{2^\lambda\!-\!1}{\overline{\rho}}}\!+\!h_s\!+\!\mu }{\sqrt{2} \sigma }\!\right)\!\right)}{2}  \nonumber\\ 
	& +\frac{\left(\overline{\rho} \sigma  (h_s+\mu)+\sqrt{\overline{\rho}(2^\lambda-1)} \sigma\right)e^{-\frac{\left(\sqrt{2^\lambda-1}-\sqrt{\overline{\rho}} (h_s+\mu)\right)^2}{2 \overline{\rho} \sigma^2}}}{\sqrt{2 \pi }}  
	\nonumber\\ 
	&+\frac{\overline{\rho} \left((h_s+\mu )^2+\sigma^2\right)\,\text{erfc}\left(\frac{\sqrt{2^\lambda-1}-\sqrt{\overline{\rho}} (h_s+\mu)}{\sqrt{2\overline{\rho}} \sigma}\right)}{2}.   
\end{align}
Here, $\text{erf}(x)=\frac{2}{\sqrt{\pi}}\int_{0}^{x}e^{-t^2}\,dt$
and $\text{erfc}(x)=1-\text{erf}(x)$.  
% To do so, we approximate the term $|\mathbf{f}_{k}(n)|^T|\mathbf{g}_{k}(n)|$  
% as a Gaussian random variable with mean $\mu_{k}$ and $\sigma_{k}^2$ variance, i.e., $|\mathbf{f}_{k}(n)|^T|\mathbf{g}_{k}(n)|\sim\mathcal{N}(\mu_{k},\sigma_{k}^2)$, by utilizing the Central Limit Theorem (CLT), which is valid for a moderately large number $M$ of reflecting elements~\cite{Galappaththige2020}. 

We can use the closed-form expression to design an iterative algorithm for calculating the average throughput $\lambda^*$, where $\epsilon$ denotes numerical accuracy. This is outlined in Theorem~2:
%\begin{theory}
%	The generated sequence $\{\lambda_l\}, l=1,2,...,\infty$ converges to the maximal average throughput $\lambda^*$. And by following Algorithm \ref{Algorithm1}, ${\mathscr K}^*$, ${\mathscr K}_k^*$, $\zeta_{k,J}, \eta_{k,J}, \forall k\in{\mathscr K}^*, \forall J\in{\mathscr K}_k^*$ are uniquely generated.
%\end{theory}
\begin{theory}
	The sequence $\{\lambda_l\}, l=1,2,...,\infty$ generated by $\lambda_{l+1}=\lambda_l+\alpha (\frac{1}{K}\sum\limits_{k=1}^K\int\limits_{0}^{+\infty}
	\max\big\{(\log_2(1+\overline{\rho}h_s^2)-\lambda^*)(\tau_d-\tau_{M,1}),0,
	\overline{\Lambda}_{k}(\lambda^*,h_s)
	\big\} d F_{|h_{k}|}(h_s)- \lambda_{l} \tau_o)$		
	with step-size $\alpha$ satisfying $\epsilon\le\alpha\le \frac{2-\epsilon}{\tau_o+\tau_d-\tau_{M,1}}$, converges to the maximal average throughput $\lambda^*$. %And by following Algorithm \ref{Algorithm1}, ${\mathscr K}^*$, ${\mathscr K}_k^*$, $\zeta_{k,J}, \eta_{k,J}, \forall k\in{\mathscr K}^*, \forall J\in{\mathscr K}_k^*$ are uniquely generated.
\end{theory}
\begin{IEEEproof}
	%The proof is similar to \cite{Zhang_IoTaas2019}.
	The convergence of $\{\lambda_l\}, l=1,\cdots,\infty$ to $\lambda^*$ can be directly proven based on the Lipschitz continuity condition presented in Proposition~1.2.3 in~\cite{Berts1999}. 
\end{IEEEproof}

\vspace{-2mm}
\subsection{Distributed Algorithm in Low Complexity}\label{optimal_strategy2}
Using Theorem~1, we can calculate the function $\overline{\Lambda}_{w(n)}(\lambda^*,|h_{w(n)}(n)|)$ based on the instantaneous channel gain $h_{w(n)}(n)$ and make decisions. To reduce calculation complexity further, we derive a decision criteria under a loose condition for the maximal average throughput $\lambda^*\ge 1$. %\com{what is the final complexity?}

\begin{lemma}\label{l:form_decision}
	For each pair $\text{S}_{k}$-$\text{D}_k$, it never opts for the {\tt assist RIS} decision if $\overline{\Lambda}_{k}(\lambda^*,\sqrt{2^{\lambda^*}-1})\le 0$. Otherwise, the decision to {\tt assist RIS} can be optimal.
\end{lemma}
\begin{proof}
%	\com{you have or omitted?}See Appendix~\ref{a:form_decision}.
 It is proved by monotonicity of $\overline{\Lambda}_{k}(\lambda^*,h_s)$ over $h_s$.
\end{proof}	
%$(\log_2(1+\overline{\rho}h_{s}^2)- \lambda^*) (\tau_d-\tau_{M,1})-\overline{\Lambda}_{k}(\lambda^*,h_s)$ and 
\vspace{-0.1cm}
%To refine the first level decision of the optimal MAC strategy $\mathbf{a}_s^*$, we analyze the optimal decision conditions. 

%Based on the first-level decision structure, 
%\end{remark}

%The lemma ensures that the optimal strategy \com{never} makes the {\tt assist RIS} decision. 

%Based on the lemma result, 
Then, we define source-destination pairs set as
${\mathscr K}^*\overset{}{=}\big\{k\in\{1,...,K\}:\overline{\Lambda}_{k}(\lambda^*,\sqrt{2^{\lambda^*}-1})\!>\! 0\big\}$.
%The set ${\mathscr K}^*$ can be calculated offline, as per its definition.
%The set divides all user pairs into two categories: pairs that can benefit from RIS assistance and those that cannot. Based on Lemma \ref{l:form_decision} and Theorem \ref{th:optimal_rule1}, the optimal decision for pairs $k\notin {\mathscr K}^*$ can be expressed in an either-or form: if $R_{k,d}(n)\ge \lambda^*$, $\text{D}_{k}$ {\tt stops}; otherwise, $\text{D}_{k}$ {\tt continues}. The set ${\mathscr K}^*$ can be calculated offline, as per its definition.
Using $\overline{\Lambda}_{k}(\lambda^*,h_s)$ and $\lambda^*$, we can calculate the set ${\mathscr K}^*$ offline. 

Moreover, for all pair $k\in {\mathscr K}^*$, we define $\zeta_{k}$ as the solution to $\overline{\Lambda}_{k}(\lambda^*,h_s)=0$, and $\eta_{k}$ as the solution to $\big(\log_2(1+\overline{\rho}h_{s}^2)- \lambda^*\big) (\tau_d-\tau_{M,1})=\overline{\Lambda}_{k}(\lambda^*,h_s)$. And they exist uniquely.
%Based on the monotonicity of $\hat{\Lambda}_{k}(\lambda^*,h_s)$, 

%we can determine the solution $\zeta_{k}$ uniquely. Similar to the derivations in \com{(\ref{equ:derivative1})}, we can prove that if $\overline{\rho}h_s^2> 2^{\lambda^*}-1$, then $\big(\hat{\Lambda}_{k}(\lambda^*,h_s)-(\tau_d-\tau{M,1})(\log_2(1+\overline{\rho}h_s^2)-\lambda^*)\big)$ monotonically decreases with respect to $h_s$. Therefore, the solution $\eta_{k}$ also 

% Furthermore, by using $\hat{\Lambda}_{k}(\lambda^*,h_s)$ and $\lambda^*$, we can calculate set ${\mathscr K}^*$. Then, based on the monotonicity of $\hat{\Lambda}_{k}(\lambda^*,h_s)$, the solution $\zeta_{k}$ is uniquely determined. Similar to the derivations in (\ref{equ:derivative1}), if $\overline{\rho}h_s^2> 2^{\lambda^*}-1$, we can prove that $\big(\hat{\Lambda}_{k}(\lambda^*,h_s)-(\tau_d-\tau_{M,1})(\log_2(1+\overline{\rho}h_s^2)-\lambda^*)\big)$ monotonically decreases with respect to $h_s$. Thus, the solution $\eta_{k}$ also exists uniquely. 

Furthermore, we are now designing a network implementation algorithm based on the proposed MAC strategies to enable distributed network operation. The algorithm follows the channel access procedure outlined in Section~\ref{sub:protocol}. The algorithm operates for each source-destination pair, utilizing ${\mathscr K}^*$ and offline-obtained parameters $\zeta_{k},\eta_{k},k\in{\mathscr K}^*$. With its pure-threshold structure, the proposed algorithm achieves significantly lower online complexity at $O(1)$, in contrast to the strategy outlined in Theorem~\ref{th:optimal_rule1}, which ensures efficient decision. %fast convergence.
%\vspace{-1cm}
\begin{algorithm}[b]
	\caption{Pure-threshold algorithm by pair $k$ }\label{Algorithm3}
	\SetKwInOut{Input}{Input}\SetKwInOut{Output}{Output}
	\Input  { ${\lambda^*},{\mathscr K}^*$}
	\Repeat{all data transmissions finish}{
		\vspace{-1mm}
		After a successful channel contention, ${\text{S}}_k$ wins the channel and ${\text{D}}_k$ obtains $h_k$.
		\uIf{$k\in {\mathscr K}^*$}{
			\uIf{$|h_k|\geq  \eta_k$}{
				${\text{S}}_k$ transmits to ${\text{D}}_k$ in direct link with rate $R_{k,d}$ and duration $\tau_d-\tau_{M,1}$.}
			\uElseIf{ $|h_k|\le  \zeta_k$}{
				\vspace{-1mm}
				${\text{D}}_k$ gives up and all sources re-contend.}
			\Else{ \vspace{-1mm}
				${\text{D}}_k$ estimates RIS channels, obtains $\mathbf{f}_{k}\astrosun$ $\mathbf{g}_{k}$ and calculates $\mathbf{\Phi}_{k}^*$ and $R_{k,r}$\\			
				\uIf{$R_{k,r} \ge {\lambda^*}$}{
					\vspace{-1mm}
					${\text{S}}_k$ transmits to ${\text{D}}_k$ under aids of the RIS with rate $R_{k,r}$, RIS beamforming matrix $\mathbf{\Phi}_{k}^*$ and duration $\tau_d-\tau_{M,2}$.}
				\Else{\vspace{-2mm}
					it gives up and all sources re-contend.}
			}
		}
		\Else{
			\uIf {$R_{k,d}\ge \lambda^*$}{
				${\text{S}}_k$ transmits in direct link with rate $R_{k,d}$ and duration $\tau_d-\tau_{M,1}$.}
			\Else{
				it gives up and all sources re-contend.}
		}
	}
\end{algorithm}
\vspace{-0.0cm}

\vspace{0cm}
\section{Numerical Results}\label{s:simu}
In this section, we numerically evaluate the performance. The network consists of $K=8$ with a central frequency of $f_c=2$\,GHz, antenna gains $G_t=G_r=0$\,dBi, $\beta_0=-30$\,dB, $P_t=30$\,dBm, and $N_0=-80$\,dBm. The channel contention parameters  are $p_0=0.3$, $\delta=25\mu$s, $\tau_{R}=\tau_{C}=50\mu$s, and $\tau_p=500\mu$s. The path loss exponents are $\alpha_{1}=3$ and $\alpha_{2}=2.5$. Locations in 2D plane of $\text{S}_k$ and $\text{D}_k$ are represented by $\textbf{S}=[(0,0),(0,10),...,(0,(K-1)10)]$ and $\textbf{D}=[(150,0),(150,10),...,(150,(K-1)10)]$, respectively.
The RIS is located at $(75,100)$ with $M=32$. In Fig.~\ref{fig:comparison1a} and Fig.~\ref{fig:comparison2a}, the results, marked as ``{Analytical}", and the approximated results, marked as ``{Approximation}", are calculated based on the proposed iterative algorithm with Monte-Carlo calculated functions in (\ref{equ:lambda_def}) and the closed-form expression in (\ref{equ:ergodic_cap}), respectively. The simulation results, marked as ``{Simulation}", are calculated from the sample average of transmitted traffic and time spent in the trails of data transmissions using Algorithm~1. In all figures, we observe that the approximation, analytical, and simulation results are well-matched, which validates the accuracy of the analysis for our proposed strategy.
\begin{figure}[t!]
	\begin{center}
		\includegraphics[scale=.41]{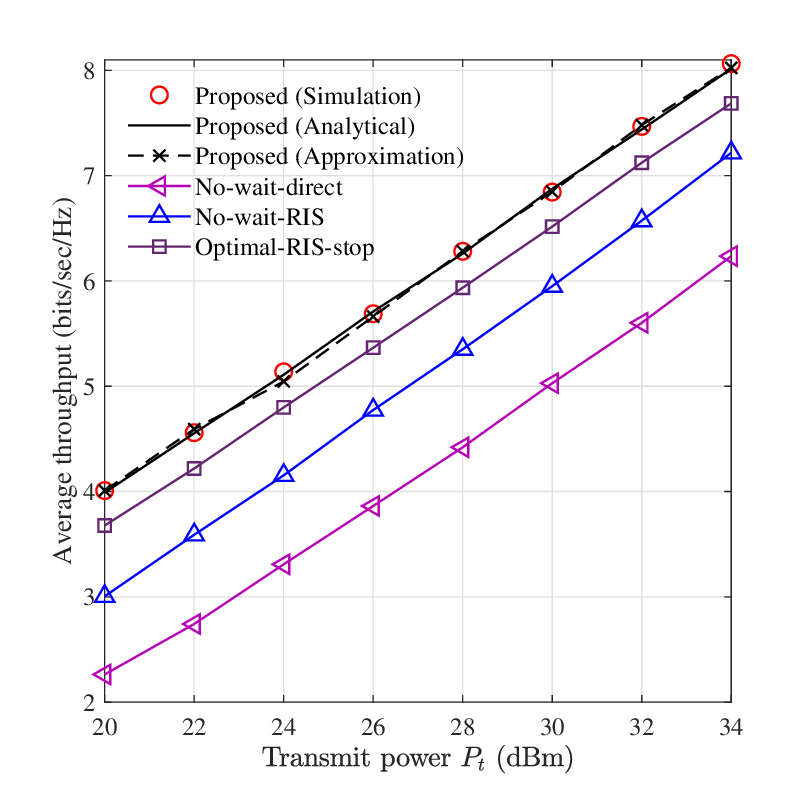}\vspace{-4mm}
		\caption{Average throughput vs $P_t$ for $K=8$ and $\tau_d=15$\,ms. }\label{fig:comparison1a}
	\end{center}
  \vspace{-0.7cm}
\end{figure}

We compare our proposed strategy with three benchmark strategies: (i) ``No-wait-direct": a winner pair has the direct link CSI and always transmits via the direct link;  (ii) ``No-wait-RIS": a winner pair has both direct and RIS link CSI, and always transmits via RIS assisted strategy, and (iii) ``Optimal-RIS-stop": a winner pair has both direct and RIS link CSI, and an optimal stopping strategy is used as in \cite{Wei2020acm}. Fig.~\ref{fig:comparison1a} shows the average throughput vs transmit power $P_t$ for $\tau_d=15$\,ms. Our proposed strategy outperforms all other strategies achieving over $8$\% average throughput gain over the best alternative strategy (optimal-RIS-stop) and over $66$\% and $27$\% throughput advantage over the no-wait-direct and no-wait-RIS strategies, respectively. 
Fig.~\ref{fig:comparison2a} plots the average throughput versus channel coherence time for the proposed and alternative strategies. The proposed strategy consistently outperforms all other strategies, with throughput advantages of more than 10\%, 31\%, and 13\% over the optimal-RIS-stop, no-wait-direct, and no-wait-RIS strategies, respectively, at $\tau_d=5$\,ms. 
% At $\tau_d=25$ ms, the proposed strategy achieves throughput advantages of 56\% and 25\% over the no-wait-direct and no-wait-RIS strategies, respectively. 
The proposed strategy exploits varying channel conditions of different user pairs and the RIS, resulting in better channel access decisions and higher average throughput gain compared to all no-wait strategies.

\vspace{-1mm}
\section{Conclusion}\label{s:con}\vspace{-1mm}
This research presents an optimal RIS-assisted channel access strategy for a distributed cooperative network involving multiple source-destination pairs and a single RIS. We use CSMA/CA strategy designed through sequential planned decision theory in a distributed manner to maximize average system throughput, rigorously proving its optimality. We achieve low online-complexity at $O(1)$ and develop a pure-threshold algorithm for distributed strategy implementation. Our proposed strategy surpasses existing ones in terms of system throughput. These findings have wide applicability in developing innovative RIS-integrated designs for opportunistic access across various networks like cellular networks and IoT networks.
Future work will extend this strategy to distributed cooperative networks with multiple source-destination pairs and multiple RISs, allowing us to select a set of RISs as an additional design parameter. 
% This research proposes an optimal RIS-assisted channel access strategy for a distributed cooperative network with multiple source-destination pairs and an RIS. The CSMA/CA strategy is developed using sequential planned decision theory in a distributed manner to maximize the average system throughput, and its optimality is rigorously proved. A low online-complexity at $O(1)$ and pure-threshold algorithm is then developed for implementing the strategy distributively. Compared to existing strategies, the proposed strategy outperforms them in terms of system throughput. The findings can be applied to develop novel RIS-integrated designs for opportunistic access in various networks such as cellular networks, IoT networks and more. Future work will extend the proposed strategy to a distributed cooperative network with multiple source-destination pairs and multiple RISs, enabling the selection of a set of RISs as an additional design parameter. This extension holds promise for unlocking new possibilities to enhance network performance and efficiency.
\begin{figure}[t!]
	\begin{center}
		\includegraphics[scale=.41]{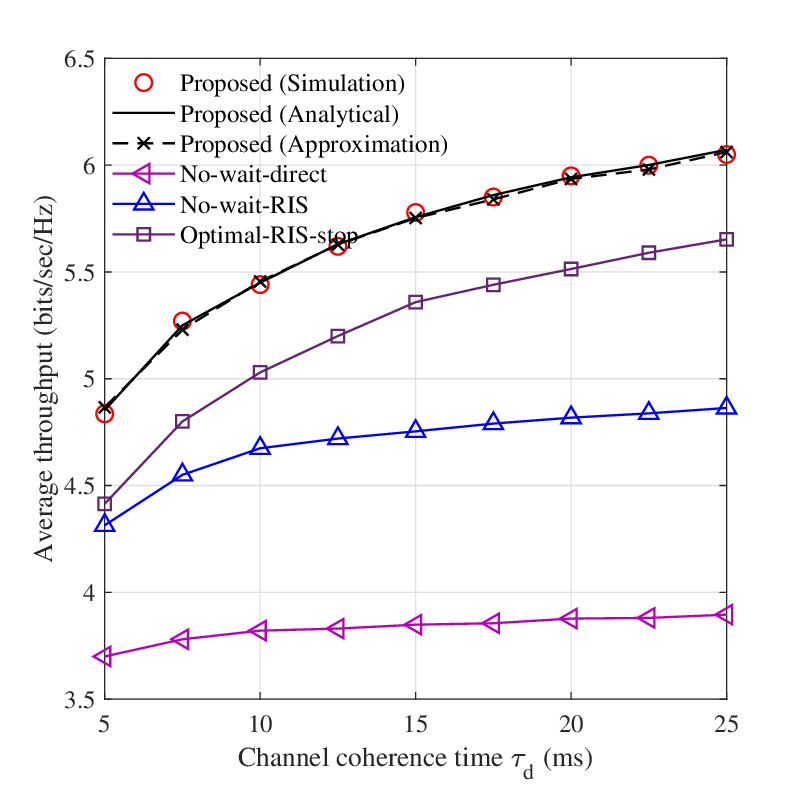}\vspace{-4mm}
		\caption{Average throughput vs $\tau_d$ for $K=8$ and $P_t=26$\,dBm.}\label{fig:comparison2a}
	\end{center}
 \vspace{-0.7cm}
\end{figure}

\vspace{-0.1cm}
\bibliographystyle{IEEEtran}
\bibliography{reference,IEEEabrv}

\end{document}